\documentclass[12pt, reqno]{amsart}
\usepackage{amssymb}
\usepackage{palatino}
\usepackage{amsmath}
\usepackage{amssymb}
\usepackage{amsthm}
\usepackage{url}

\newcommand{\Mod}[1]{\ (\mathrm{mod}\ #1)}

\newcommand\al\alpha
\newcommand\be\beta
\newcommand\de\delta
\newcommand\la\lambda
\newcommand\tha\theta

\newcommand\iy\infty

\newcommand\bma{\begin{pmatrix}}
\newcommand\ema{\end{pmatrix}}

\newtheorem{theo}{Theorem}

\newtheorem{pro}{Proposition}
\newtheorem{Proposition}{Proposition}
\everymath{\displaystyle}

\title{$C_\lambda$- Extended oscillator algebra and  $d$-orthogonal  polynomials}
\author{Fethi Bouzeffour $^{\ast, \diamond}$, Wissem Jedidi $^{\S, \diamond}$}
\address{$^\ast$ Department of mathematics, College of Sciences, King Saud University,  P. O Box 2455 Riyadh 11451, Saudi Arabia.}
\address{$^\S$ Department of Statistics \& OR, King Saud University, P.O. Box 2455, Riyadh 11451, Saudi Arabia}
\address{$^\diamond$ Universit\'e de Tunis El Manar, Facult\'e des Sciences de Tunis, LR11ES11 Laboratoire d'Analyse Math\'ematiques et Applications, 2092, Tunis, Tunisia.} \email{fbouzaffour@ksu.edu.sa; wissem.jedidi@fst.utm.tn}

\begin{document}
 \maketitle

{\bf Abstract:} In this paper we first construct an analytic  realization of the $C_\lambda$-extended oscillator algebra with the help of difference-differential operators.
 Secondly, we study families of  $d$-orthogonal polynomials which are extensions of the Hermite and Laguerre polynomials. The underlying algebraic framework allowed us a systematic derivation of their main properties such as recurrence relations, difference-differential equations, lowering and rising operators and generating functions.  Finally,  we use these polynomials to construct a realization of the $C_\lambda$-extended oscillator by block matrices.
\vskip0.2cm
{\bf Keywords}: vector orthogonal polynomials, deformed oscillator algebra .
\section{Introduction}
In \cite{Quesne}, Quesne defined the $C_\lambda$-extended oscillator algebra ($\lambda$ being an integer bigger than $1$) as the associative algebra generated by four elements $a_-,$ $a_+,$ $N$, $s$  that obey the commutation relations
\begin{align} \label{s}
&[a_-,\,a_+] = 1+\sum_{i=1}^{\lambda-1}\,\beta_i \; s^i,\qquad s^{\lambda} = 1, \\
&[N,\,a_-] = - a_-,\qquad [N,\,a_+] = a_+,\qquad [N,\,s]=0,\\
&a_-s = \varepsilon_\lambda\; s \; a_- \qquad a_+ s =\varepsilon_\lambda^{-1}\; s \; a_+,\qquad \varepsilon_\lambda=e^{\frac{2i\pi}{\lambda}},\\
&N^* = N,\qquad a_+^*=a_-,\qquad s^*=s^{-1},
\end{align}
where the constants structure $\beta_1,\dots,\,\beta_{\lambda-1}$ are complex numbers restricted by the conditions
\begin{equation}
\overline{\beta_i}=\beta_{\lambda-i},\quad i=1,\,\dots, \lambda.\label{cc1}
\end{equation}
The connection between orthogonal polynomials and harmonic oscillator as well as the quantum algebra is well known \cite{Vilenkin, Koor}. The oscillator algebra of creation, annihilation, and number operators plays a central role in the investigation of many class of orthogonal polynomials, and provides a useful tool to derive  their operational properties such as recurrence relations, difference equation, lowering and rising operators and generating functions. It particular,
when  $\lambda=2$, the $C_\lambda$-extended Heisenberg algebra is reduced to the Calogero--Vasiliev algebra \cite{8},  which is given
 by the generators  $a_-,\,a_+,\, R$, $1$ that satisfy  the commutation relations
\begin{eqnarray*}
[a_-,\,a_+] &=& 1+2\nu R,\qquad R^2=1,\\
\{a_\pm,R\}&=& 0,\qquad [1,\,a_\pm]= [1,\,R]=0.
\end{eqnarray*}
A realization of the Calogero--Vasiliev algebra working in the
Schr\"{o}dinger representation, $\Psi =\Psi(x)$,  is given by the operators
\begin{equation}\label{H4}
a_-=\frac{1}{\sqrt{2}}(Y_\nu +x),\quad a_+=\frac{1}{\sqrt{2}}(-Y_\nu+x),
\end{equation}
where  $Y_\nu$ is the {\em Dunkl operator} (corresponding  to the root system $A_1$, see \cite[Definition 4.4.2)]{18}), which is a differential-difference
operator, depending on a parameter $\nu\in \mathbb{R}$ and acting on $C^\infty(\mathbb{R})$ as follows:
\begin{equation}
Y_\nu :=\frac{d}{dx}+\,\frac{\nu}x(1-R),
\label{72}
\end{equation}
where $R$ is the Klein operator acting on function $\psi$ of the real variable $x$ as
\begin{equation}
(R\psi)(x)=\psi(-x)
\end{equation}
(see \cite{8} for more details).
Note that the operator $Y_\nu$ is also related by a simple similarity transformation to the Yang-Dunkl operator used in \cite{8}, were the authors show that the Calogero--Vasiliev algebra is intimately related to parabosons and parafermions, and to the $osp (1|2)$ and $osp (2|2)$ superalgebras. \\

The  Hamiltonian  takes the form
\begin{align}
H_\nu&=-\frac{1}{2}Y_\nu^2+\frac{1}{2}x^2=-\frac{1}{2}\frac{d^2}{dx^2}-\frac{\nu}x\,\frac{d}{dx}\,+\frac{\nu}{2x^2}
(1-R)+\frac{1}{2}x^2
\end{align}
and the  wave functions corresponding to the well-known eigenvalues
\begin{equation}
\lambda_n=n+\nu+\frac{1}{2},\quad n=0,\,1,\,2,\dots
\end{equation}
are given by
\begin{equation*}
\psi^{(\nu)}_n(x)= \gamma_n\,  e^{-x^2/2}H^{(\nu)}_n(x),
\end{equation*}
where
\begin{equation}  \label{constant3}
\gamma_n= \left(2^{2n}\; \Gamma([\frac{n}{2}]+1) \;\Gamma([\frac{n+1}{2})]+\nu+\frac{1}{2}) \right)^{-\frac{1}{2}}
\end{equation}
and  $H^{(\nu)}_n(x)$ is the generalized Hermite polynomial introduced by Szeg\"{o} \cite{z} and obtained in \cite{z} from Laguerre polynomials  $L_n^\nu(x)$  by the means:
 \begin{equation} \begin{cases}
     H_{2n}^{(\nu)}(x) = (-1)^n \,2^{2n} \, n!\,\, L_n^{\nu-\frac{1}{2}}(x^2),\\
     H_{2n+1}^{(\nu)}(x) =(-1)^n \, 2^{2n+1}\, n! \,\,x\,L_n^{\nu+\frac{1}{2}}(x^2).
   \end{cases}\label{hermout1}\end{equation}

It is well known that for $\nu>-\frac{1}{2},$ these polynomials satisfy the orthogonality relations:
\begin{equation}
\int_{\mathbb{R}}H^{(\nu)}_n(x)
H^{(\nu)}_m(x)|x|^{2\nu}e^{-x^2}\,dx= \frac{1}{\gamma_n^2}\,\delta_{n\,m},\label{orth}
\end{equation}
where $\gamma_n$ is given in \eqref{constant3}. Many of the known generalized Hermite polynomials  are also the eigenfunctions of the energy operator for a deformed oscillator (see \cite{Flor, Askey}).  \\

In this paper, we discuss the connection of some class of $d$-orthogonal polynomials
with the $C_\lambda$-extended oscillator algebra (for $\lambda=d+1$).  The $d$-orthogonal polynomials can be obtained from general multiple orthogonal polynomials under some restrictions upon their parameters (see \cite{Van Assche}). Applications of  $d$-orthogonal polynomials include the simultaneous Pad\'e approximation problem where the   multiple orthogonal polynomials appear  (see \cite{Beukers}).  Also they play important role in random matrix theory  (see \cite{Akemann}). The $d$-orthogonal polynomials have been intensively studied in the last 30 years due to their intriguing properties and applications (see \cite{Ben,Douak} and further references in the literature).\\

One problem that deserves attention is to relate $d$-orthogonal polynomials to some oscillator algebras. In  \cite{zhe,bouz1,bouz2}, the authors have found some $d$-orthogonal polynomials related to the deformed harmonic oscillator and share a number of properties  with the classical orthogonal polynomials. In this paper, we discuss the connection of some class of $d$-orthogonal polynomials with the $C_\lambda$-extended oscillator algebra $(\lambda=d+1)$. We use a Klein type operator $S$ of finite order to construct  a realization of the $C_\lambda$-extended oscillator algebra in terms of  difference-differential operator and a system  of vector orthogonal polynomials obtained from $d$-orthogonal polynomials, in order to provide realizations of the $C_\lambda$-extended oscillator by block matrices. Note that the obtained $d$-orthogonal polynomials  are eigenfunctions of the operator $S$.\\

The paper is organized as follows. In Section 2, we   review the definition of the $C_\lambda$-extended oscillator algebra  and  introduce a Bergmann realization associated  to this oscillator. In section 3, we deal with the formalism of the exponential of the Dunkl type operators, which are the most commonly used operators in the context of evolution problems. We establish the relevant rules  to the action of an exponential operator on a given function and those  for the disentanglement of exponential operators. We also establish there the most important properties of $d$-orthogonal polynomials. In section 4,  we construct a family of vector orthogonal polynomial and  analyse the effect  of the generators of the $C_\lambda$-extended algebra on it.
\section{Bergmann realization of the $C_\lambda$-extended oscillator algebra}
\subsection{The $C_\lambda$-extended oscillator algebra }

The element $s$ introduced in \eqref{s},  is a generator of the cyclic group
\begin{equation}\label{vovov}
C_\lambda=\{1,\,s,\,s^2,\,\dots,\,s^{\lambda-1}\}
\end{equation}
of order $\lambda$.  The primitive idempotents related to this group are  denoted by $\Pi_0,\,\Pi_1,\dots,\,\Pi_{\lambda-1}$ and are  given by
\begin{equation}
\Pi_{i}=\frac{1}{\lambda}\sum_{j=0}^{\lambda-1}\varepsilon_\lambda^{-ij}s^j,\quad i=0,\dots,\lambda-1,\qquad \varepsilon_\lambda=e^{\frac{2i\pi}{\lambda}}\label{projection}
\end{equation}
and the bosonic $C_\lambda$-extended oscillator Hamiltonian is defined by
\begin{equation}
H=\frac{1}{2}\{a_-,a_+\}.
\end{equation}
According to  \cite{QuesneVQ}, it  may therefore be rewritten in terms of $a_-,$ $a_+,$ $N$, $\Pi_0,\dots,\,\Pi_{\lambda-1}$ as follows:
\begin{align}
&[a_-,\,a_+]=1+\sum_{j=0}^{\lambda-1}\widehat{\beta}_j\Pi_j,\\& [N,\,a_-]=-a_-,\quad
[N,\,a_+]=a_+,\quad a_+ \Pi_i=\Pi_{i+1}a_+,\\&
\sum_{i=1}^{\lambda}\Pi_i =1,\quad \Pi_i\Pi_j=\delta_{ij}\Pi_i,\quad
\Pi_i^\dagger=\Pi_i,
\end{align}
where, according to condition \eqref{cc1}, the discrete Fourier transforms $\widehat{\beta}_0,\dots \widehat{\beta}_{\lambda-1},$  are
real numbers given by
$$ \widehat{\beta}_j=\sum_{i=0}^{\lambda-1}\varepsilon_\lambda^{ij}\beta_i,\,\, j=1,\dots,\lambda-1, $$
and are restricted by the condition  $\sum_{i=0}^{\lambda-1}\widehat{\beta}_i=0$.\\

Next proposition will be needed in the sequel and gives a relation in the $C_\lambda$-extended oscillator algebra.
\begin{pro}\label{pro1}Let $n\in \mathbb{N},$ we have\\
\noindent 1. $[a_-^n,a_+]=\big(n+\sum_{i=0}^{\lambda-1}\beta_i\frac{\varepsilon_\lambda^{ni}-1}{\varepsilon_\lambda^{i}-1}s^i\big)\, a_-^{n-1},$\\
\noindent 2. $[a_-,a_+^n]=a_+^{n-1}\, \big(n+\sum_{i=0}^{\lambda-1}\beta_i\frac{\varepsilon_\lambda^{ni}-1}{\varepsilon_\lambda^{i}-1}s^i\big),$\\
\noindent 3. $[N,a_-^n]= -na_-^{n-1}, \qquad [N,a_{+}^n]=na_{+}^{n-1}.$\\

\noindent In particular, it holds that
\begin{align*}
&[a_-^{n\lambda},a_+]=n\lambda a_-^{n\lambda-1}  \quad\text{and} \text \quad [a_-,a_+^{n\lambda}]=n\lambda a_+^{n\lambda-1}.
\end{align*}
\end{pro}
\begin{proof}
The proof is easy and follows from the identity
\begin{align*}
[a^n,b]=\sum_{i=0}^{n-1}a^{n-1-i}[a,b]a^i.
\end{align*}
\end{proof}
According to  \cite{QuesneVQ}, the $C_\lambda$-extended oscillator algebra possesses a canonical irreducible representation defined on the orthonormal basis $|n\rangle,\,\, n=0\,1,\,2,\,.\dots$, endowed with the following actions:
\begin{align}
&N|k\lambda+i\rangle=(k\lambda+i)|k\lambda+i\rangle,\quad
\Pi_j|k\lambda+i\rangle=\delta_{ij}|k\lambda+i\rangle,\label{r22}\\&
a_+|k\lambda+i\rangle=\sqrt{k\lambda+i+1+\widehat{\beta}_{i+1}}|k\lambda+i+1\rangle,\label{r24}\\&
a_-|k\lambda+i\rangle=\sqrt{k\lambda+i+\widehat{\beta}_i}|k\lambda+i-1\rangle.\label{r23}
\end{align}
\subsection{Dunkl type operator}
Let ${\mathrm{\nu}}=(\nu_1,\dots,\nu_{\lambda-1})$, where $\nu_1,\dots,\nu_{\lambda-1}$ are complex numbers satisfying the condition
\begin{equation}
\sum_{l=0}^{\lambda-1} \nu_l=0.
\end{equation}
Consider the  Klein-type reflection  $S$  acting on functions $f(z)$ of complex variable $z$ as
follows
\begin{equation}(Sf)(z):=f(\varepsilon_\lambda z).\label{ref111}
\end{equation}
and its associated  differential-difference operator
\begin{equation}
Y_{\mathbf{\nu}}:=\frac{d}{dz}+\frac{1}{z}\sum_{j=1}^{\lambda-1}\nu_j\,S^j. \label{diffh13}
 \end{equation}
Recall that the complex reflection group $G(r,1,N)$ consists of the $N \times N$ permutation matrices with the nonzero entries being powers of $\varepsilon=e^{\frac{2i\pi}{r}}.$ The group $G(r,1,N)$ is also generated by the transpositions $(i, i + 1),\,\, i=1,\,\dots,\, N-1,$ and by the complex reflections $S_i$   defined by
$$S_iz=(z_1,\,\dots,\varepsilon z_i,\,\dots,\,z_N).$$
When  $N=1$, it happens that:
\begin{itemize}
\item  the group $G(r,1,N)$ is isomorphic to the cyclic group $C_\lambda$ defined in \eqref{vovov};
\item the operator defined in \eqref{diffh13} is a particular case of the complex Dunkl operator $Y_i$ associated to $G(r,1,N)$ given by  \cite{18}:
\begin{align}
Y_i=\frac{\partial}{\partial z_i}+\kappa_0\sum_{j\neq
i}\sum_{j=0}^{r-1}\frac{1-S_i^{-j}(i,j)S_i^j} {z_i-\varepsilon^j
z_j}+\sum_{j=1}^{r-1}\kappa_j\sum_{l=0}^{r-1}\frac{\varepsilon^{-rl}s_i^l}{z_i},\label{Dunkl1}
\end{align}
where $\kappa_0,\,\kappa_1,\dots,\kappa_r$ are real numbers.
\end{itemize}
\subsection{Some results on deformed factorial numbers}
Let $\nu_i, \, i=0, \, 1, \cdots, \, \lambda -1$ be a sequence  of complex numbers and  $\widehat{\nu}_s$ their discrete Fourier transforms defined as
\begin{equation}
\widehat{\nu}_s=\sum_{l=0}^{\lambda-1} \nu_l\varepsilon_\lambda^{sl}.
\end{equation}
One can notice that
\begin{equation}
\widehat{\nu}_n =\widehat{\nu}_s,\quad  \text {if} \quad n= s\Mod{\lambda}.
\end{equation}
Then, for nonnegative integers $n$, we introduce the deformed numbers $[n]_{\mathbf{\nu}}$ and the deformed factorial numbers by
\begin{align}
&[n]_{\mathbf{\nu}}= n+\widehat{\nu}_n,  \label{tt}\\
&[0]_{\mathbf{\nu}} != 1,\quad [n+1]_{\mathbf{\nu}}!=[n+1]_{\mathbf{\nu}}\;[n]_{\mathbf{\nu}}!. \label{fact}
\end{align}
For
\begin{equation}\label{ak}
\alpha_k=\frac{k+\widehat{\nu}_k}{\lambda}, \quad k=1,\,2,\,\dots,\,\lambda-1.
\end{equation}
and for $0\leq s\leq \lambda-1$, we need to introduce the multi-index numbers
\begin{equation}
 \Delta(\nu,s)=\left\{
  \begin{array}{l l}
    (1,\,\alpha_1,\,\dots,\,\alpha_{\lambda-1})\quad \mbox{if}\quad s=0,\\
   (1,\,\alpha_1+1,\,\dots,\,\alpha_s+1,\,\alpha_{s+1},\dots,\alpha_{\lambda-1})\quad \mbox{otherwise}.
  \end{array} \right.  \label{S1}
\end{equation}
We obtain the following result:
\begin{pro} The initial problem
\begin{equation}
 \left\{
  \begin{array}{l l}
    Y_{\mathbf{\nu}} f(z)=\varrho f(z),\\
   f(0)=1
  \end{array} \right.  \label{S1}
\end{equation}
has a unique analytic solution given by
\begin{align}
\mathcal{E}_\lambda(\varrho z,{\mathbf{\nu}})=\sum_{n=0}^{\infty}\frac{(\varrho z )^{n}}{[n]_{\mathbf{\nu}}!}. \label{E1}
\end{align}
Furthermore, the generalized exponential function $\mathcal{E}_\lambda(\varrho\,.\,,\nu)$ has the following hypergeometric representation
\begin{align}\label{H21}
\mathcal{E}_\lambda(z,\nu)=\sum_{s=0}^{\lambda-1}\frac{( z/\lambda)^s}{[s]_\nu!}\,_{0}F_{\lambda-1}\left(\left. \begin{matrix} {-}\\
   \Delta(\nu,s)\end{matrix}\right\vert (\frac{z}{\lambda})^\lambda\right).
\end{align}
\end{pro}
\begin{proof}
It is easily seen from the action of the operator $Y_\nu$ on monomials $z^n$:
\begin{equation}
Y_\nu z^n=[n]_\nu\,z^{n-1},\label{Dz}
\end{equation}
that the function $\mathcal{E}_\lambda(\varrho\,.\,,\nu)$ is the unique solution of the system \eqref{E1}.
To prove \eqref{H21}, it suffice to write  $[n\lambda+s]_\nu!\,$ in terms of the Pochhammer Symbol $(a)_n$ defined for $a \in \mathbb{C}$ and $n=1,\,\dots\,$ by
$$(a)_0 :=1,\quad(a)_n:=a(a+1)\dots (a-n+1).$$
For $s=1,\,2,\,\dots \lambda-1,$ and with $\alpha_k$ defined in \eqref{ak}, we get
\begin{align*}
[n\lambda+s]_\nu!&=\prod_{l=1}^{n}(l\lambda+\widehat{\nu}_0) \prod_{k=1}^s\prod_{l=0}^{n}\big(l\lambda+k+\widehat{\nu}_k\big)
\prod_{k=s+1}^{\lambda-1}\prod_{l=0}^{n-1}\big(l\lambda+k+\widehat{\nu}_k\big),\\
&= \lambda^{n\lambda+s}\,n!\,\prod_{k=1}^s(\frac{k+\widehat{\nu}_k}{\lambda})_{n+1} \prod_{k=s+1}^{\lambda-1}(\frac{k+\widehat{\nu}_k}{\lambda})_{n}\\&
=\lambda ^{n\lambda+s}\,n!\,\prod_{k=1}^s(\alpha_k)_{n+1} \prod_{k=s+1}^{\lambda-1}(\alpha_k)_{n}\\&=
\lambda ^{n\lambda+s}\,n!\,\prod_{k=1}^s\frac{\Gamma(\alpha_k+n+1)}
{\Gamma(\alpha_k)}\prod_{k=s+1}^{\lambda-1}\frac{\Gamma(\alpha_k+n)}{\Gamma(\alpha_k)}.
\end{align*}
Similarly,
\begin{align*}
[n\lambda]_\nu!=
\lambda^{n\lambda}\,n!\prod_{k=1}^{\lambda-1}(\alpha_k)_{n}=\lambda^{n\lambda}\,n!\prod_{k=1}^{\lambda-1}\frac{\Gamma(\alpha_k+n)}{\Gamma(\alpha_k)}.
\end{align*}
If we decompose the sum \eqref{E1}  in the form
\begin{align}
\mathcal{E}_\lambda(z,\nu)&=\sum_{s=0}^{\lambda-1}\sum_{n=0}^{\infty}\frac{( zt )^{n\lambda+s}}{[n\lambda +s]_\nu!},
\end{align}
we obtain
\begin{align}
\mathcal{E}_\lambda(z,\nu)=\sum_{s=0}^{\lambda-1}\frac{( zt/\lambda)^s}{[s]_\nu!}\,_{0}F_{\lambda-1}\left(\left. \begin{matrix} {-}\\
  \alpha_1+1,\dots,\alpha_s+1,
  \alpha_{s+1},\dots,\alpha_{\lambda-1}\end{matrix}\right\vert (\frac{z}{\lambda})^\lambda\right),
\end{align}
and this allows to conclude.
\end{proof}
\subsection{Bergmann realization}
Let $ G^{m ,n}_{p, q }$ be  the Meijer's $G$-function \cite{Luke},  $m_s$ be the measure  defined by
\begin{equation}
dm_s(z)=\frac{\lambda^{s+1}}{\prod_{k=1}^{\lambda-1}\Gamma(\alpha_k)}G^{\lambda , 0}_{0, \lambda
  } \left(\frac{r^{2\lambda}}{\lambda^\lambda}\, \Bigg|
      \begin{array}{c}
        -\\[0.1cm]
         \Delta(\nu,s)
      \end{array}\right)dr\,d\theta,\quad z=re^{i\theta},
\end{equation}
and $L_{\nu,\lambda}^2(\mathbb{C})$ be the space of measurable functions $f$ on $\mathbb{C}$ that satisfy
\begin{equation}
\|f\|^2 _{\nu,\lambda}=
\sum_{s=0}^{\lambda-1}\int_{\mathbb{C}}|\pi_s(f)(z)|^2\,dm_s(z)< \infty,
\end{equation}
The generalized Bergmann space $\mathfrak{B}_{\nu,\lambda}(\mathbb{C})$ is the pre-Hilbert space of analytic functions in $L_{\nu,\lambda}^2(\mathbb{C})$, equipped with the inner product
\begin{equation}  \label{e9}
\left\langle f,g\right\rangle _{\nu}=
\sum_{s=0}^{\lambda-1}\int_{\mathbb{C}}\pi_s(f)(z)\overline{\pi_s(g)(z)} \,dm_s(z),
\end{equation}
where \begin{equation}
\pi_s=\frac{1}{\lambda}\sum_{l=0}^{\lambda-1}\varepsilon_\lambda^{-sl}S^l.
\end{equation}
\begin{theo}If $f,g\in \mathfrak{B}_{\nu,\lambda}(\mathbb{C}) $ expanded in the form
$$f(z)=\sum_{n=0}^\infty a_n z^n\quad  \mbox{and} \quad g(z)=\sum_{n=0}^\infty b_nz^n,$$
then, their inner product is given by
\begin{equation}
\left\langle f,g\right\rangle _{\nu}=\sum_{n=0}^\infty a_n\overline{b}_n [n]_\nu!.\label{pro}
\end{equation}
\end{theo}
\begin{proof}Let $0\leq s,s'\leq \lambda-1$ and $n,\,m=0,\,\dots. $ From \eqref{pro}, it is clear that for $s\neq s'$, we have   $$\pi_l(z^{n\lambda+s})\overline{\pi_l(z^{n\lambda+s'})}=0\quad \mbox{for} \quad l=0,\dots\lambda-1,$$
hence,
\begin{align*}
\left\langle z^{n\lambda+s},z^{m\lambda+s'}\right\rangle _{\nu}=0.
\end{align*}
For $s=s'$, we have
\begin{align*}
\left\langle z^{n\lambda+s},z^{m\lambda+s}\right\rangle _{\nu}&=\frac{1}{\pi} \frac{\lambda^{s+1}}{\Pi_{k=1}^{\lambda-1}\Gamma(\alpha_k)}
\int_{0}^\infty r^{(n+m)\lambda+2s}  G^{\lambda , 0}_{0, \lambda} \left(\frac{r^{2\lambda}}{\lambda^\lambda}\, \Bigg|
      \begin{array}{c}
        -\\[0.1cm]
         \Delta(\nu,s)
      \end{array}\right)\,dr
\\ & \quad \times \int_{0}^{2\pi}e^{i\lambda\theta(n-m)}\,\ \,d\theta.
\end{align*}
Now using the following  well known  formula \cite{Luke}
\begin{align}
&\int_{0}^\infty x^{n-1} G^{r , s}_{p,q
  } \left(\frac{x}{a}\, \Bigg|
      \begin{array}{c}
        \beta_1,\dots,\,\beta_p\\[0.1cm]
         \gamma_1,\dots,\,\gamma_q
      \end{array}\right)\,dx\nonumber
      \\&=a^n\,\frac{\prod_{i=1}^r\Gamma(\gamma_i+n)\prod_{i=1}^s\Gamma(1-\beta_i-n)}
      {\prod_{i=r+1}^q\Gamma(1-\gamma_i-n)\prod_{i=s+1}^p\Gamma(\beta_i+n)}, \label{integ}
\end{align}
we obtain
\begin{align*}
\left\langle z^{n\lambda+s},z^{m\lambda+s}\right\rangle _{\nu}& =\delta_{nm}\lambda^{n\lambda+s }\,n!\,\prod_{k=1}^s(\alpha_k)_{n+1}
\prod_{k=s+1}^{\lambda-1}(\alpha_k)_{n}\\&=\delta_{nm}[n\lambda+s]_\nu!.
\end{align*}
\end{proof}
It is oblivious that the space $\mathfrak{B}_{\nu,\lambda}(\mathbb{C})$ is a Hilbert space equipped with the inner product \eqref{e9} and the monomials $$\{e_n(z)=z^n/\sqrt{[n]_\nu!},\,\; n=0,\,1,\,2.\dots\}$$
constitute an orthonormal basis for this space. The space $\mathfrak{B}_{\nu,\lambda}(\mathbb{C})$ has the kernel function
\begin{equation}
K_\nu(w,z)=\sum_{n=0}^\infty e_n(z)e_n(\overline{w})=\mathcal{E}_\lambda(z\overline{w},\nu).
\end{equation}

\begin{pro}$(i)$ Let  $Z$ be the multiplication operator  ($(Zf)(z):=zf(z)$.  The operators $Y_\nu$ and $Z$  are closed densely defined operators on the space $\mathfrak{B}_{\nu,\lambda}(\mathbb{C})$ on the common  domain
\begin{equation*}
D=\{f(z)=\sum_{n=0}^\infty a_nz^n\,:\,\sum_{n=0}^\infty |a_n|^2 [n+1]_\nu<\infty\}.
\end{equation*}
Furthermore, for $f,g \in D$ one has
\begin{equation*}
\left\langle Y_\nu f,g\right\rangle _{\nu}=\left\langle f,Zg\right\rangle _{\nu}.
\end{equation*}
\end{pro}
\begin{pro}
The following holds:\\
$(i)$  $$S^*=S^{-1},\quad Y_\nu S=\varepsilon_\lambda SY_\nu,\quad SZ=\varepsilon_\lambda ZS.$$
$(ii)$
\begin{align}
&[Y^n_\nu,z]=\big(n+\sum_{i=1}^{\lambda-1}\nu_i(\varepsilon_\lambda^{in}-1)S^i\big)Y_\nu^{n-1}
\\&[Y_\nu,z^n]=z^{n-1}\big(n+\sum_{i=1}^{\lambda-1}\nu_i(\varepsilon_\lambda^{in}-1)S^i\big).
\end{align}
\end{pro}
Now, assume that the complex numbers $\nu_i$ are restricted by the conditions
\begin{equation}
\nu_{\lambda-i}=-\varepsilon_\lambda^i\nu_i,\quad i=1,\dots ,\lambda-1.
\end{equation}
The case where the structure constants $\beta_i$ in \eqref{s} take the values
$$\beta _i=\nu_i(\varepsilon_\lambda^i-1),\quad i=1,\dots ,\lambda-1,$$
we get a convenient one variable model of the representation of $C_\lambda$-extended oscillator algebra given by the Hilbert space $\mathfrak{B}_{\nu,\lambda}(\mathbb{C})$ and the action of the $C_\lambda$-extended oscillator algebra is given by
\begin{align} (sf)(z) := f[\varepsilon z],\,\,\,\,
(af)(z):=(Y_\nu f)(z),\,\,\,\, (a_+ f)(z):=Zf(z) \end{align}
and the related bosonic extended Hamiltonian takes the form
\begin{equation}
H=z\frac{d}{dz}+\frac{1}{2}+\frac{1}{2}\sum_{j=1}^{\lambda-1}\nu_j\,(\varepsilon_\lambda^i+1)S^i.
\end{equation}
\section{Generalized Hermite polynomials }
For fixed integer $m$, Gould and Hopper \cite{Gould} defined the generalized Hermite polynomials $H_n(x,m)$ by the operational identity
\begin{equation*}
H_n(x,m)=e^{(\frac{d}{dx})^m}\,x^n,\quad n=0,\,1,\,\dots,\,.
\end{equation*}
Similarly, a  new family of generalized  Hermite polynomials $ \{H^{(\lambda,\nu)}_n(x)\}_{n=0}^\infty$ can be determined by means of the following operational formula
\begin{align}
&H^{(\lambda,\nu)}_n(x)=e^{-Y^\lambda_\nu/\lambda}x^n.\label{d2}
\end{align}
Due to the fact that $Y_\nu$ reduces the degrees of the polynomials , see equation \eqref{Dz} then the latter series  \eqref{d2} contains only a finite number of terms. The operational definition \eqref{d2} greatly simplifies the study of the generalized Hermite
polynomials. From Proposition 4, we deduce that
\begin{align}
[Y_\nu^{n\lambda},x]=n\lambda x^{n\lambda-1},
\quad[Y_\nu,x^{n\lambda}]=n\lambda Y_\nu^{n\lambda-1}
\end{align}
and that
\begin{align}
[e^{\mu Y_\nu^\lambda},\,x]=\mu \lambda Y_\nu^{\lambda-1} e^{\mu Y_\nu^\lambda},
\quad[Y_\nu,\,e^{\mu x^\lambda}]=\mu\lambda x^{\lambda-1}e^{\mu x^\lambda}.\label{recc2}
\end{align}
\begin{Proposition}The following holds true
\begin{align*}
&Y_\nu H^{(\lambda,\nu)}_n(x)=[n]_\nu\,H^{(\lambda,\nu)}_{n-1}(x),\\& (x-Y_\nu^{\lambda-1})H^{(\lambda,\nu)}_{n}(x)=H^{(\lambda,\nu)}_{n+1}(x).
\end{align*}
\end{Proposition}
\begin{proof}
Performing the {\it Dunkl operator} $Y_\nu$
of both sides of \eqref{d2} with respect to $x$, we obtain
\begin{equation}\label{f1}
Y_\nu H^{(\lambda,\nu)}_n(x)=[n]_\nu\,H^{(\lambda,\nu)}_{n-1}(x).
\end{equation}
The equation \eqref{recc2} can be exploited to derive the operational formula
\begin{align*}
e^{-Y_\nu^\lambda/\lambda}x^n&=e^{-Y_\nu^\lambda/\lambda}xe^{Y_\nu^\lambda/\lambda}e^{-Y_\nu^\lambda/\lambda}x^{n-1}= (x-Y_\nu^{\lambda-1})e^{-Y_\nu^\lambda/\lambda}x^{n-1}.
\end{align*}
Hence \begin{equation}\label{d1}
H^{(\lambda,\nu)}_{n}(x)=(x-Y_\nu^{\lambda-1})H^{(\lambda,\nu)}_{n-1}(x).
\end{equation}
\end{proof}
We also obtain
\begin{theo}The polynomials $H^{(\lambda,\nu)}_{n}(x)$ satisfy the following higher order  differential-difference equations
\begin{align*}
&Y_\nu (x-Y^{\lambda-1}_\nu) H^{(\lambda,\nu)}_{n}(x)=[n+1]_\nu H^{(\lambda,\nu)}_{n}(x),\\&
(x-Y^{\lambda-1}_\nu) Y_\nu H^{(\lambda,\nu)}_{n}(x)=[n]_\nu H^{(\lambda,\nu)}_{n}(x).
\end{align*}
\end{theo}
\begin{proof}
Applying the operator $Y_\nu$ to the both side of \eqref{d1}, we get the following differential-difference equation:
\begin{equation}\label{diff}
Y_\nu (Y^{\lambda-1}_\nu- x) H^{(\lambda,\nu)}_{n}(x)=-[n+1]_\nu H^{(\lambda,\nu)}_{n}(x).\end{equation}
\end{proof}
\begin{Proposition}The polynomials $H^{(\lambda,\nu)}_{n}$ are generated by the series
\begin{equation*}
e^{-t^\lambda/\lambda}\mathcal{E}_\lambda(xt,\nu)=\sum_{n=0}^\infty
H^{(\lambda,\nu)}_{n}(x)\frac{t^n}{[n]_\nu!}.
\end{equation*}
Furthermore,
\begin{equation}\label{H1}
H^{(\lambda,\nu)}_n(x)=\sum_{k=0}^{[\frac{n}{\lambda}]}
\frac{(-1)^k[n]_\nu!}{\lambda^{\lambda k}
\,k!\,[n-k\lambda]_\nu!}\,x^{n-k\lambda}.
\end{equation}
\end{Proposition}
\begin{proof}
According to \eqref{d1} and \eqref{S1}, we can write
\begin{align*}
\sum_{n=0}^\infty
H^{(\lambda,\nu)}_{n}(x)\frac{t^n}{[n]_\nu!} =\sum_{n=0}^\infty
e^{-Y_\nu^\lambda/\lambda}x^n\frac{t^n}{n!} =e^{-Y_\nu^\lambda/\lambda}
\mathcal{E}_\lambda(xt,\nu) =e^{-t^\lambda/\lambda}\mathcal{E}_\lambda(xt,\nu)
\end{align*}
Hence, the  polynomials $\{H^{(\lambda,\nu)}_{n}(x)\}$ are generated by
\begin{equation}
e^{-t^\lambda/\lambda}\mathcal{E}_\lambda(xt,\nu)=\sum_{n=0}^\infty
H^{(\lambda,\nu)}_{n}(x)\frac{t^n}{[n]_\nu!}.\label{gene}
\end{equation}
An explicit formula for the generalized Hermite polynomials $H^{(\lambda,\nu)}_n(x)$ is obtained by  expanding in power series the generation function given in \eqref{gene}:
\begin{align*}
e^{-t^\lambda/\lambda }\mathcal{E}_\lambda(xt,\nu)&=\sum_{k=0}^\infty \frac{(-1)^k}{\lambda^{ k}}\frac{t^{k\lambda}}{k!}\sum_{m=0}^\infty
\frac{(xt)^m}{[m]_\nu!}\\&=\sum_{k=0}^\infty \sum_{m=0}^\infty
\frac{(-1)^k}{\lambda^{k}k!}\frac{x^m}{[m]_\nu!}t^{k\lambda+m}.
\end{align*}
Substituting $n=k\lambda+m$, then $0\leq k\leq [n/\lambda]$  and we
\begin{align}
e^{-t^\lambda/\lambda}\mathcal{E}_\lambda(xt,\nu)=\sum_{n=0}^\infty \sum_{k=0}^{[\tfrac{n}{\lambda}]}
\frac{(-1)^{k}}{\lambda^{k}k!}\frac{x^{n-k\lambda}}{[n-k\lambda]_\nu!}t^{n}.\label{egen1}
\end{align}
Finally,
\begin{equation}\label{H1}
H^{(\lambda,\nu)}_n(x)=\sum_{k=0}^{[\frac{n}{\lambda}]}
\frac{(-1)^k[n]_\nu!}{\lambda^{\lambda k}k![n-k\lambda]_\nu!}x^{n-k\lambda}.
\end{equation}
\end{proof}
\begin{Proposition}
The polynomials $H^{(\lambda,\nu)}_{n}$ satisfy the following three terms recurrence relations
\begin{equation}
xH^{(\lambda,\nu)}_{n}(x)=H^{(\lambda,\nu)}_{n+1}(x)+[n]_\nu\dots [n-\lambda]_\nu H^{(\lambda,\nu)}_{n-\lambda+1}(x).\label{Rebouh}
\end{equation}
\end{Proposition}
\begin{proof}
By taking into account \eqref{f1} and \eqref{d1} we get the three terms recurrence relations
\begin{equation*}
xH^{(\lambda,\nu)}_{n}(x)=H^{(\lambda,\nu)}_{n+1}(x)+[n]_\nu\dots [n-\lambda]_\nu H^{(\lambda,\nu)}_{n-\lambda+1}(x).
\end{equation*}
\end{proof}
\section{Matrix realization of the $C_\lambda $-extended oscillator}
In sequel we assume that $d=\lambda-1$. We propose here a realization  of the $C_\lambda$-extended oscillator by matrices.  \\

Let us first clarify the relationship  between the notion of d-orthogonality of a family of polynomials  and  orthogonality  of vector polynomials.  Recall that $d$-orthogonal polynomials are system $\{P_n(x)\}$ of monic polynomials (with  $\deg P_n = n$)  such that there exists a {\em  vector linear functional} $\,\mathcal{U}=
\begin{bmatrix}
u_{0}  \dots  u _{d-1}
\end{bmatrix}
^{T}$ satisfying the conditions:
\begin{equation}\label{orthd}
\begin{array}{rll}
i) & \langle u_j,x^kP_n\rangle=0 ,& 0\leq k\leq [\frac{n-j}{d}] \\
ii) &  \langle u_j,x^nP_{nd+j}\rangle\neq 0\,, & n=0,\,\dots,\,.
\end{array}%
\end{equation}
where $\langle u,P\rangle$ is the effect of a linear functional $u$ on a polynomial $P$ and $[x]$
denotes the greatest integer function. Note that the case $d = 1$ corresponds to the ordinary notion of orthogonal
polynomials.  According to \cite{Maroni, Iseghem} the vector orthogonality relations is equivalent to the existence of a linear $(d+2)$-term recurrence relation
\begin{equation}
xP_n(x) = P_{n+1} + \sum_{j=0}^{d}a_j (n)P_{n-j}(x) \label{d-ort}
\end{equation}
with constants   $a_j(n)$,  $a_j(d) \neq 0$.

\begin{itemize}
\item  A system of polynomials $\{P_{n}\}$ is said to be $d $\textit{--symmetric} when it verifies%
\begin{equation}
P_{n}(\varepsilon _dx)=\varepsilon_{d+1}^{n}P_{n}(x)\,,\ \ n\geq 0\,,\quad \varepsilon_{d+1}=e^{\frac{2i \pi}{d+1}}.
\label{[Sec2]-d-symmetric-Sn}
\end{equation}%
\item  We say that the vector of linear functionals $\,\mathcal{U}%
=%
\begin{bmatrix}
u _{1},  \,\dots ,\, u _{d}%
\end{bmatrix}%
^{T}$ is said to be $d$\textit{--symmetric} when the moments of its entries
satisfy, for every $n\geq 0$,%
\begin{align}
&(u_j)_{(n+1)d+k-1}:=\langle u _{j},x^{(d+1)n+k-1 }\rangle=0\\&
1\leq j\leq d \quad 1\leq k \leq d+1,\quad j\neq k .  \label{[Sec2]-momL1}
\end{align}
\end{itemize}
According to \cite{DM-JAT-82},  for every sequence of monic polynomials $\{P_{n}\}$, $d$--orthogonal with respect to the \textit{vector of linear functionals } $\mathcal{U}=
\begin{bmatrix}
u_{0}  \dots  u _{d-1}
\end{bmatrix}
^{T}$, the following statements are equivalent :
\begin{itemize}
\item[$(i)$] the vector of linear functionals $\mathcal{U}$ is $d$%
--symmetric;
\item[$(ii)$] the sequence $\{P_{n}\}$ is $d$--symmetric;
\item[$(iii)$] the sequence $\{P_{n}\}$ satisfies
\begin{equation}
xP_{n+d}(x)=P_{n+d+1}(x)+\gamma _{n+1}P_{n}(x),\label{chikha}
\end{equation}%
with $P_{n}(x)=x^{n}$ for $0\leq n\leq d$.
\end{itemize}

Within this context, we obtain the following result on the generalized Hermite polynomials:
\begin{theo}For $d=\lambda-1$, the family $\{H^{(\lambda,\nu)}_{n}(x)\}$ are $d$-orthogonal polynomials with respect to the functionals $u_0,\,\dots,\,u_{d-1}$, which are  determined by their moments:
\begin{equation}
\langle u_k,x^{n\lambda+s}\rangle=\delta_{ks}\int_0^\infty u^{n\lambda+s}\upsilon_s(u)\,du=\frac{[n\lambda+s]_\nu!}{n!\,\lambda^n},
\end{equation}
where \begin{equation}
\upsilon_s(u)=\frac{\lambda u^{-s-1}}{\prod_{i=1}^{\lambda-1}\Gamma(\alpha_i)}G^{\lambda-1 ,0}_{0,\lambda-1
  } \left(\frac{u^\lambda}{\lambda^{\lambda-1}}\, \Bigg|
      \begin{array}{c}
        -\\[0.1cm]
          \Delta(\nu,s)
      \end{array}\right).
\end{equation}
\end{theo}
\begin{proof}From \eqref{Rebouh}, the polynomials $\{H^{(\lambda,\nu)}_{n}(x)\}$ satisfy a $(\lambda+1)$-term recursion relation of the form  \eqref{chikha} and then, are $d$-orthogonal $d$-symmetric  polynomials.  Furthermore  there exist $\lambda -1$  symmetric functionals $u_k, k =0,\ldots,d-1$, on the space of all polynomials  $\mathcal{P}$ such that
\begin{equation}\label{tictac}
\begin{cases}
u_k (P_mP_n)    = 0, m <[\frac{n-k}{d}], \\
u_k(P_nP_{n(d+1)+k})    \neq 0,\, n   \geq 0.
\end{cases}
\end{equation}
To determine the moments of the  functional $u_k$, we will need the following inversion formula
\begin{equation}\label{H1}
x^m=\sum_{n=0}^{[\tfrac{m}{\lambda}]}
\frac{[m]_\nu!}{\lambda^{n}n![m-n\lambda]_\nu!}H^{(\lambda,\nu)}_{m-n\lambda}(x).
\end{equation}
To prove  \eqref{H1}, we can use the generating function \eqref{gene} to first  obtain
\begin{equation}
\mathcal{E}_\lambda(xt,\nu)=\sum_{n,k=0}^\infty
\frac{H^{(\lambda,\nu)}_{k}(x)}{\lambda^{n}n![k]_\nu!}t^{n\lambda+k},\label{ge}
\end{equation}
then, substituting $m=n\lambda+k$, then $0\leq n\leq [m/\lambda]$,  we get
\begin{align}
\mathcal{E}_\lambda(xt,\nu)=\sum_{m=0}^\infty \sum_{n=0}^{[\tfrac{m}{\lambda}]}
\frac{H^{(\lambda,\nu)}_{m-n\lambda}(x)}{\lambda^{n}n![m-n\lambda]_\nu!}t^{m}.\label{g2}
\end{align}
and it remains to compare the coefficient of $t^m$ in the two expansion \eqref{E1} and \eqref{g2} of $\mathcal{E}_\lambda(xt,\nu)$.\\
Following Maroni \cite{Maroni},  the moments of the  functionals $u_k, k =0,\ldots,d-1$ related to $\{H^{(\lambda,\nu)}_{n}(x)\}$ can be obtained from the inversion formula \eqref{H1}
\begin{equation}
(u_k)_{n\lambda+s}=\langle u_k,x^{n\lambda+s}\rangle=\delta_{ks}\frac{[n\lambda+s]_\nu!}{n!\lambda^{n}}.\label{mm}
\end{equation}
It remains  to prove the integral representation of the functionals $u_k$, so it suffices to substitute
 \begin{align*}&p=s=0,\quad q=r=\lambda-1,\\&\gamma_i=\alpha_i+1, i=1,\dots, s,\quad  \gamma_i=\alpha_i,\quad i=s+1,\dots \lambda-1,\end{align*} in the integral \eqref{integ}, in order to get
 \begin{align*}
\int_{0}^\infty x^{n-1} G^{\lambda-1 ,0}_{0,\lambda-1
  } \left(\frac{x}{\lambda^{\lambda-1}}\, \Bigg|
      \begin{array}{c}
        -\\[0.1cm]
         \Delta(\nu,s)
      \end{array}\right)\,dx\nonumber\
     =\lambda^{-s}\prod_{i=1}^{\lambda-1}\Gamma(\alpha_i)\frac{[n\lambda+s]_\nu!}{n!\lambda^n}.
\end{align*}
The latter yields
\begin{equation}
\int_0^\infty u^{n\lambda+s}\upsilon_s(u)\,du=\frac{[n\lambda+s]_\nu!}{n!\lambda^n},
\end{equation}
where \begin{equation}
\upsilon_s(u)=\frac{\lambda u^{-s-1}}{\prod_{i=1}^{\lambda-1}\Gamma(\alpha_i)}G^{\lambda-1 ,0}_{0,\lambda-1
  } \left(\frac{u^\lambda}{\lambda^{\lambda-1}}\, \Bigg|
      \begin{array}{c}
        -\\[0.1cm]
          \Delta(\nu,s)
      \end{array}\right),
\end{equation}
so that
\begin{equation}
\langle u_k,x^{n\lambda+s}\rangle=\delta_{ks}\int_0^\infty u^{n\lambda+s}\upsilon_s(u)\,du=\frac{[n\lambda+s]_\nu!}{n!\lambda^n}.
\end{equation}
\end{proof}
 Let $\{P_{n}\}$ be a $d$-symmetric
family of $d$-orthogonal polynomials with respect  to the vector of functionals $\,\mathcal{U}=\big[u_0,\dots,u_{d-1}\big]^T$ . We define the family of  vector polynomials  $\{\mathbb{P}_{n}\}$ related to $\{P_{n}\}$  as follows :
\begin{equation}\mathbb{P}_n=\big[P_{nd},\dots,P_{(n+1)d-1}\big]^T \end{equation}
and we extend the action of the vector of functionals $\,\mathcal{U}$  in vector polynomials $\{\mathbb{P}_n\}$ as follows:
\begin{equation*}
\mathcal{U}(\mathbb{P})=%
\begin{bmatrix}
u_{0}(P_{1}) & \cdots & u_{d-1}(P_{1}) \\
\vdots & \ddots & \vdots \\
u^{0}(P_{d}) & \cdots & u_{d-1}(P_{d})%
\end{bmatrix}%
\end{equation*}
Within
this context, $\{\mathbb{P}_{n}\}$ is said to be a vector
orthogonal polynomial sequence with respect to the vector of functionals $%
\mathcal{U}$, if
\begin{equation}
\left\{
\begin{array}{rll}
i) & (x^{k}\mathcal{U})(\mathcal{P}_{n})=0_{d\times d}\,, & k=0,1,\ldots
,n-1\,, \\
ii) & (x^{n}\mathcal{U})(\mathcal{P}_{n})=\Delta _{n}\,, &
\end{array}%
\right.
  \label{Orth}
\end{equation}
where $\Delta _{n}$ is a regular upper triangular $d\times d$ matrix.\\

We can see easily that the  $d$-orthogonality of a family of polynomials $\{P_n\}$ defined in \eqref{orthd} is equivalent to the vector orthgonality of the related family $\{\mathbb{P}_n\}$ defined as above.  For a deeper account of the
theory (in a more general framework, considering quasi--diagonal
multi--indices) we refer the reader to \cite{Branquinho}.\\

\indent  Let consider the  family of vector polynomials  $\{\mathbb{H}^{(\lambda,\nu)}_n(x)\}$  given by
\begin{equation}
\mathbb{H}^{(\lambda,\nu)}_n(x)=\big[\widetilde{H}^{(\lambda,\nu)}_{nd}(x),\,\dots,\, \widetilde{H}^{(\lambda,\nu)}_{(n+1)d-1}(x)\big]^T,\quad n\in \mathbb{N}.
\end{equation}
where
\begin{equation}
\widetilde{H}^{(\lambda,\nu)}_{n}(x)=([n]_\nu!)^{-1/2}H^{(\lambda,\nu)}_{n}(x).
\end{equation}
From \eqref{Rebouh},  we easily check  the three terms recurrence relations
\begin{equation}
x\widetilde{H}^{(\lambda,\nu)}_{n}(x)=\sqrt{[n+1]_\nu}\widetilde{H}^{(\lambda,\nu)}_{n+1}(x)+\alpha_n \widetilde{H}^{(\lambda,\nu)}_{n-\lambda+1}(x).\label{Rebouh1}
\end{equation}
where \begin{equation}
\alpha_n=\sqrt{\frac{[n]_\nu!}{[n-\lambda+1]_\nu!}}.
\end{equation}
We also extend  the action of the operator $Y_\nu$ on the vector polynomial  $\mathbb{H}^{(\lambda,\nu)}_n(x)$
through
$$  Y_\nu \mathbb{H}^{(\lambda,\nu)}_n(x)=\big[Y_\nu \widetilde{H}^{(\lambda,\nu)}_{nd}(x),\,\dots,\, Y_\nu  \widetilde{H}^{(\lambda,\nu)}_{(n+1)d-1}(x)\big]^T. $$

\begin{theo}Under the above notations, we have
\begin{align}
&x\mathbb{H}^{(\lambda,\nu)}_n=A_n\mathbb{H}^{(\lambda,\nu)}_{n+1}+B_n\mathbb{H}^{(\lambda,\nu)}_n+C_{n}\mathbb{H}^{(\lambda,\nu)}_{n-1}\,,\ \ n=0,1,\ldots , \label{sss}\\&
Y_\nu \mathbb{H}^{(\lambda,\nu)}_n=A^T_{n-1}\mathbb{H}^{(\lambda,\nu)}_{n-1}+B^T_n\mathbb{H}^{(\lambda,\nu)}_n\,,\ \ n=0,1,\ldots  \label{sss}\\&
S\mathbb{H}^{(\lambda,\nu)}_n=R_n\mathbb{H}^{(\lambda,\nu)}_n.
\end{align}
with $\mathbb{H}^{(\lambda,\nu)}_{-1}=
\begin{bmatrix}
0, & \cdots, & 0%
\end{bmatrix}%
^{T}$, $\mathbb{H}^{(\lambda,\nu)}_{0}=%
\begin{bmatrix}
\widetilde{H}^{(\lambda,\nu)}_{0}(x), & \cdots, & \widetilde{H}^{(\lambda,\nu)}_{d-1}(x)%
\end{bmatrix}%
^{T}$, and matrix coefficients $A_n$, $B_n$, $C_{n}, $ $R_n \in \mathcal{M}_{d\times d}$
given by
\begin{eqnarray*}
A_n &=&%
\begin{bmatrix}
0 & 0 & \cdots & 0 \\
\vdots & \vdots & \ddots & \vdots \\
0 & 0 & \cdots & 0 \\
\sqrt{[(n+1)d]_\nu} & 0 & \cdots & 0%
\end{bmatrix}%
\\  B_n&=&%
\begin{bmatrix}
0 & \sqrt{[nd+1]_\nu} &  &  \\
& \ddots & \ddots &  \\
&  & 0 & \sqrt{[nd+d-1]_\nu} \\
&  &  & 0%
\end{bmatrix}%
  \label{[Sec2]-3TTR-MCoefs} \\
C_{n} &=&\operatorname{diag}\,[\gamma _{nd},\,\gamma _{nd+1},\ldots ,\gamma _{(n+1)d-1}],\\
R&=&\operatorname{diag}\,[1,\,\varepsilon_\lambda,\ldots ,\varepsilon^{\lambda-1}_\lambda]
\notag
\end{eqnarray*}
\end{theo}
The position operator $x$ can be identified with the following  Bloc matrix 
\begin{equation*}
\mathbb{X}=%
\begin{bmatrix}
B_0 & C_1 &  &  &  \\
A_{0} & B_1 & C_2 &  &  \\
& A_{1} & B_2 & C_3 &  \\
&  & \ddots & \ddots & \ddots%
\end{bmatrix}%
\end{equation*}
and the momentum operator $Y_\nu$ with
\begin{equation*}
\mathbb{Y}=%
\begin{bmatrix}
B^T_0 & A^T_0 &  &  &  \\
0 \,& B^T_1 & A^T_1 &  &  \\
&\, 0\, & B^T_2 & A^T_2 &  \\
&  & \ddots & \ddots & \ddots%
\end{bmatrix}%
\end{equation*}

\end{document}